\documentclass[manuscript,screen,nonacm]{acmart} % Preprint version
\usepackage{fixme}
\usepackage{xcolor}
\usepackage{subfig}

\fxsetup{layout=inline,status=final} % Compile with status=final for no notes
\fxsetface{inline}{\color{olive}}

\AtBeginDocument{%
  \providecommand\BibTeX{{%
    \normalfont B\kern-0.5em{\scshape i\kern-0.25em b}\kern-0.8em\TeX}}}

\newcommand*{\nn}[1]{\langle#1\rangle} % Nearest neighbors/Encoding
\newcommand*{\class}[1]{\textbf{#1}} % Complexity class

\newtheorem{problem}{Problem}
\newtheorem{claim}{Claim}
\newtheorem{fact}{Fact}

\newtheorem{repclaiminner}{Claim}
\newenvironment{repclaim}[1]{%
  \renewcommand\therepclaiminner{#1}%
  \repclaiminner
}{\endrepclaiminner}

\begin{document}

\title{Non-NP-Hardness of Translationally-Invariant Spin-Model Problems}

\author{Rotem Liss}
\email{rotemliss@cs.technion.ac.il}
\author{Tal Mor}
\email{talmo@cs.technion.ac.il}
\author{Roman Shapira}
\email{romanshap@cs.technion.ac.il}
\affiliation{%
  \institution{Technion --- Israel Institute of Technology}
  \streetaddress{Technion City}
  \city{Haifa}
  \country{Israel}
  \postcode{3200003}
}

% TODO : Check keywords
\keywords{Quantum Complexity, Local Hamiltonian Problem, Mahaney's theorem, Fortune's theorem, Sparse languages, Ising, Heisenberg, t-J, Hubbard, Fermi-Hubbard}

\begin{abstract}
Finding the ground state energy of the Heisenberg Hamiltonian is an important problem in the field of condensed matter physics.
In some configurations, such as the antiferromagnetic translationally-invariant case on the 2D square lattice, its exact ground state energy is still unknown.
We show that finding the ground state energy of the Heisenberg model cannot be an \class{NP-Hard} problem unless \class{P}=\class{NP}.
We prove this result using a reduction to a sparse set and certain theorems from computational complexity theory.
The result hints at the potential tractability of the problem and encourages further research towards a positive complexity result.
In addition, we prove similar results for many similarly structured Hamiltonian problems, including certain forms of the Ising, t-J, and Fermi-Hubbard models.
\end{abstract}

\maketitle

\section{Introduction} \label{sec:intro}
\subsection{Background and Motivation}
Quantum complexity theory classifies the difficulty of computational problems using quantum computing resources, inspired by methods from classical complexity theory.
Important quantum complexity classes include \class{BQP}, the class of problems efficiently solvable by a quantum computer, and \class{QMA}, the class of problems efficiently verifiable by a quantum computer given a quantum hint.
Most problems dealt with in the field come naturally from the physical world, with a prominent example being the $k$-local Hamiltonian (promise) problem.
The $k$-local Hamiltonian problem is deciding whether the ground state energy (lowest eigenvalue) of a given $k$-local Hamiltonian (a sum of $Poly(n)$ Hermitian matrices, each acting on at most $k$ qubits, where $n$ is the number of qubits) is below a given threshold $\alpha$ or above a given threshold $\beta$, promised one of these to be the case.
A cornerstone result of the theory is the ``Quantum Cook-Levin'' theorem by Kitaev, Shen and Vyalyi \cite{ksv02}, showing the $k$-local Hamiltonian problem to be a complete problem for the class \class{QMA}, for $k = O(\log(n))$.
Since the result of \cite{ksv02}, researchers have invested significant efforts to find more \class{QMA-Complete} problems (for a partial review, see \cite{boo14}).
Since similarly to $\class{NP} \neq \class{P}$ it is widely conjectured that $\class{QMA}\neq\class{BQP}$, such results provide evidence of said problems being likely intractable on a quantum computer.

Nevertheless, some special or restricted cases of \class{QMA-Complete} problems, that are still of physical importance, remain unclassified or turn out to be tractable on either classical or quantum computers.
As a simple example, the $k$-local Hamiltonian problem is in \class{P} for $k=1$.
Another example is the quantum $k$-SAT problem, a version of the $k$-local Hamiltonian problem that allows only projection operators for the local Hamiltonian terms.
It turns out that the problem is \class{QMA\textsubscript{1}-Complete}
\footnote{\class{QMA\textsubscript{1}} is \class{QMA} with no error on the ``yes'' instances}
for $k \geq 3$ \cite{gn16}, but it is in \class{P} for $k=2$ \cite{bra11};
thus studying restrictions of known \class{QMA-Complete} problems may bear useful or surprising results.

\subsection{The Heisenberg Hamiltonian}
In this paper, we are interested in analyzing the complexity of a specific restriction of the local Hamiltonian ground state energy problem:
the Heisenberg Hamiltonian, on an $N\times N$ two-dimensional (2D) square lattice, with a translationally-invariant (TI) interaction strength $J$ and \emph{no 1-local terms}.
The Hamiltonian is of the following form:
\begin{equation} \label{eq:heis}
    H_{\text{TI-Heisenberg}}=J\sum_{\nn{i,j}}\vec{\sigma}_{i}\cdot\vec{\sigma}_{j}=J\sum_{\nn{i,j}}(X_{i}X_{j}+Y_{i}Y_{j}+Z_{i}Z_{j}),
\end{equation}
with $\nn{i,j}$ denoting nearest neighbors on the 2D square lattice and $X, Y, Z$ denoting the Pauli matrices.
Our specific target is the antiferromagnetic case ($J>0$), as the ferromagnetic case ($J<0$) has been exactly solved (see, for example, Chapter 6.1 in \cite{faz99}) with its ground state energy scaling polynomially with $N$, placing the problem corresponding to the ferromagnetic case in \class{P}.
%Moreover, as for any $J \neq 1$ the problem instance can be reduced to the $J = 1$ instance by scaling, it may be assumed without loss of generality the latter case holds.

Despite restricting ourselves to what seems to be a narrow case, the antiferromagnetic Heisenberg Hamiltonian given by eq.~\eqref{eq:heis} is a useful and extensively studied model in condensed-matter physics; 
for example, see \cite{ll09, sdo94, lc92, md90, lin90, th89}.
Although past studies have established various properties and symmetries of the ground state, to the best of our knowledge, no closed-form solution for the ground state or its energy has been found yet.
Therefore, attempting to analyze the model as a computational problem with complexity theoretic tools may prove beneficial in uncovering its tractability.

\subsection{Related and Previous Work} \label{subsec:prev_work}
Local Hamiltonian problems based on the Heisenberg model under various constraints have been previously studied.
In \cite{sv09}, Schuch and Verstraete showed a Heisenberg Hamiltonian with additional arbitrary 1-local terms $\vec{B}_i$ (one term per each lattice site) is \class{QMA-Complete}:
\begin{equation} \label{eq:sv09}
    H_{\text{Heisenberg-1-local}}=J\sum_{\nn{i,j}}\vec{\sigma}_{i}\cdot\vec{\sigma}_{j}+\sum_{i}\vec{B}_{i}\cdot\vec{\sigma}_{i}.
\end{equation}
Then, in \cite{pm17}, Piddock and Montanaro showed a Heisenberg Hamiltonian \emph{not} restricted to a 2D lattice and with variable, edge-dependent interaction strengths $J_{ij}$, but still without 1-local terms, is \class{QMA-Complete}:
\begin{equation} \label{eq:pm17}
    H_{\text{Heisenberg}}=\sum_{i,j}J_{ij}\vec{\sigma}_{i}\cdot\vec{\sigma}_{j}.
\end{equation}
Being significant generalizations of eq.~\eqref{eq:heis}, it is thus reasonable that both Hamiltonian problems given by eqs.~\eqref{eq:sv09} and \eqref{eq:pm17} were shown to be \class{QMA-Complete}.

For the problem of eq.~\eqref{eq:heis}, however, the same result most likely does not hold.
An observation made by Cubitt and Montanaro in \cite{cm16} (Observation 30) shows the Heisenberg Hamiltonian problem of eq.~\eqref{eq:heis} is in the class \class{StoqMA} (as defined in \cite{bbt06}, see also \cite{bdot08}), as it can be reduced to a \emph{stoquastic} local Hamiltonian problem (see Appendix~\ref{sec:stoqma} for details).
It is shown by \cite{bbt06} that $\class{MA} \subseteq \class{StoqMA} \subseteq \class{AM}$, and thus \class{StoqMA} is in the polynomial hierarchy.
It is therefore unlikely the problem of eq.~\eqref{eq:heis} is in \class{QMA-Complete}, since \class{QMA} (and even \class{BQP}) is conjectured to lie outside the polynomial hierarchy \cite{aar10}.

Knowing the problem is in \class{StoqMA}, we naturally wonder whether it is also \class{StoqMA-Complete}, or in some smaller complexity class than \class{StoqMA}.
For this, we consider two related Hamiltonians discussed in Appendix E of \cite{cm16}; the fact their corresponding problems are both in \class{P} hints that our Hamiltonian problem (of eq.~\eqref{eq:heis}) may, too, be tractable.
These are the Heisenberg Hamiltonian on the complete graph of $N$ vertices, and the Heisenberg Hamiltonian on the complete bipartite graph of $2N$ vertices (also known as the Lieb-Mattis model).
The exact solutions for the ground states of both Hamiltonians are shown in \cite{cm16}, with the ground state energies being a polynomial in $N$, placing both computational problems in \class{P}.
While the geometries of these Hamiltonians differ from the geometries of the 2D square lattice present in the problem of eq.~\eqref{eq:heis}, all three problems share a common trait:
it is enough to supply just $N$ and $J$ as input arguments, as the geometry of the underlying Hamiltonian may be encoded in the solving algorithm.
As both problems of \cite{cm16} lie in \class{P}, this serves as potential evidence that the problem of eq.~\eqref{eq:heis} may also be in \class{P} or at least in some tractable complexity class.

Lastly, we note that the same translationally-invariant antiferromagnetic Heisenberg Hamiltonian on a \emph{1-D chain} of $N$ spins (instead of a 2-D lattice) was also widely studied and solved by Bethe in 1931 \cite{Bet31}.
The exact ground state energy has later been found by Hulthén \cite{hul38}, and established to be $E_0=N\cdot(1/4-\ln(2))$ (in the limit $N \rightarrow \infty$)\footnote{For a derivation of the ground state energy using Bethe's solution, see \cite{Mat81}, chapter 5.10}, placing the corresponding computational problem in \class{P}.
This result serves as additional (albeit weaker\footnote{Increasing the dimension of a problem may increase its complexity; for example, the classical Ising Hamiltonian with arbitrary weights and no 1-local terms is in \class{P} in the 2-D lattice case but is \class{NP-Complete} in the 3-D case \cite{bar82}}) evidence that our model of interest could be tractable.

\subsection{Our Contribution}
In this paper, we prove the following result for the complexity of the problem of eq.~\eqref{eq:heis}:

\begin{claim} \label{clm:main}
    If $\class{P}\neq\class{NP}$, the Hamiltonian promise problem of eq.~\eqref{eq:heis} is neither $\class{NP-Hard}$ nor $\class{coNP-Hard}$.
\end{claim}

As an immediate consequence of Claim~\ref{clm:main}, we can refute the problem's membership in \class{StoqMA-Complete}, as that would imply it belongs to \class{StoqMA-Hard}, which is a subclass of \class{NP-Hard}.
In fact, the problem cannot be \class{C-Complete} for any class $\class{C} \supseteq \class{NP}$ or $\class{C} \supseteq \class{coNP}$.
This result provides evidence for the conjecture that the problem lies in some tractable complexity class, even though finding its definite solution has so far proved challenging.

The proof of Claim~\ref{clm:main} relies on two theorems in complexity theory:
Mahaney's theorem \cite{mah82} and Fortune's theorem \cite{for79}, which allow us to exclude the problem from \class{NP-Hard} and \class{coNP-Hard}, respectively.
Both theorems are formulated for decision problems (languages) and employ the notion of \emph{sparse languages}, which will be defined in Section~\ref{sec:prelim}, along with the precise definition of the Heisenberg Hamiltonian problem.

We present our results in Section~\ref{sec:res}, beginning with the proof of Claim~\ref{clm:main} in Subsections~\ref{subsec:red_to_sparse}--\ref{subsec:conc}.
We then extend the result to additional models:
in Subsection~\ref{subsec:uc-h} we show the result applies to a family of d-dimensional ``Unit-Cell" translationally-invariant Hamiltonians.
Then, in Subsection~\ref{subsec:hubb} we prove the result holds for other translationally-invariant spin-models: the Fermi-Hubbard and t-J models, relevant for quantum chemistry and high temperature superconductivity, respectively.
Lastly, in Section~\ref{sec:disc} we provide a discussion of our results, including the non-\class{NP-Hardness} of various cases of the Ising model (defined in Appendix~\ref{sec:ising}) and the implications of this finding.

% TODO : Maybe uncomment later, need to check if needed in the conference
% \subsection{Structure of this paper}
% The structure of this paper is as follows: in section~\ref{sec:prelim} we present necessary formal definitions: the problem statement, Mahaney's and Fortune's theorems, and sparse promise problems.
% In section~\ref{sec:res} we prove the promise problem given by eq.~\eqref{eq:heis} (and formally defined in section~\ref{sec:prelim}) can be reduced to a sparse promise problem according to our previous definition.
% With this fact and by formulating promise problem versions of Mahaney's and Fortune's theorems we conclude the main result - Claim~\ref{clm:main}.
% In section~\ref{sec:disc} we discuss several implications of our result and possibilities for \emph{positive} results on the complexity of the problem, concluding with a short summary.

\section{Preliminaries} \label{sec:prelim}
We begin this section by providing a formal definition of the main computational (promise) problem we analyze in this work, based on eq.~\eqref{eq:heis}, and discuss several subtleties of the definition.
We continue by defining sparse languages and sparse promise problems.
We conclude the section by presenting Mahaney's and Fortune's theorems.

\subsection{Problem Statement}

\begin{problem}[Two-dimensional (2D) translationally-invariant (TI) Heisenberg Hamiltonian]
\label{prob:heis}
    The problem is defined as follows:
    \begin{itemize}
        \item Inputs: \begin{itemize}
            \item $N$ : 2D square lattice side length (encoded in unary)
            \item $J$ : translationally-invariant interaction strength (a fixed-precision real number, encoded in binary) such that $0 < J < Poly(N)$
            \item $\alpha, \beta$ : energy threshold arguments (fixed-precision real numbers, encoded in binary) such that $\beta-\alpha > 1/Poly(N)$
        \end{itemize}
        \item Goal: decide whether $\lambda(H_{\text{TI-Heisenberg}}) < \alpha$ (a ``Yes'' instance) or $\lambda(H_{\text{TI-Heisenberg}}) > \beta$ (a ``No'' instance), promised one of these to be the case, where:
        \begin{equation} \tag{\ref{eq:heis}}
            H_{\text{TI-Heisenberg}}=J\sum_{\nn{i,j}}\vec{\sigma}_{i}\cdot\vec{\sigma}_{j}=J\sum_{\nn{i,j}}(X_{i}X_{j}+Y_{i}Y_{j}+Z_{i}Z_{j}),
        \end{equation}
        with $\nn{i,j}$ denoting nearest neighbors on the 2D square lattice, $X, Y, Z$ denoting the Pauli operators, and $\lambda(H)$ denoting the ground state energy of a Hamiltonian $H$.
    \end{itemize}
\end{problem}

Following are several remarks about the definition:

\paragraph{Fixed-precision real numbers}
The arguments $J, \alpha, \beta$ are represented in binary as fixed-precision real numbers.
We assume each fixed-precision real number $A$ is represented by an integer part of $P = f(N)$ bits and a fractional part of $Q = g(N)$ bits for some functions $f, g$, such that:
\begin{equation}
    A = \sum_{i=0}^{P-1}A^{\text{Int}}_{i}2^i + \sum_{i=1}^{Q}\frac{A^{\text{Frac}}_{i}}{2^i},
\end{equation}
where $A^{\text{Int}}_{i}, A^{\text{Frac}}_{i} \in \left\{-1,0,1\right\}$ are the $i$'th \emph{signed} bits of the integer and fractional parts of $A$, respectively.

\paragraph{Topology of the Hamiltonian}
The Hamiltonian given in Problem~\ref{prob:heis} is on a \emph{planar} 2D square lattice with open boundary conditions.
However, our results also hold if we put the 2D square lattice on a torus (with closed boundary conditions), where each vertex has an equal number of nearest neighbors.

\paragraph{Choice of unary encoding for $N$}
In the general Heisenberg Hamiltonian problems analyzed by \cite{sv09} and \cite{pm17} (eqs.~\eqref{eq:sv09}--\eqref{eq:pm17}), $N^2$ 1-local terms $\vec{B}_i$ and $O(N^2)$ interaction strengths $J_{ij}$ are supplied as input arguments, respectively. 
Thus in these problems, the total input length scales polynomially with the lattice side length $N$ (representing the system size).
As we are interested in achieving similar scaling in our problem, we choose $N$ to be supplied in unary\footnote{
The reader is encouraged to compare to \cite{gi13}, where a binary encoding of $N$ is applied, leading to results with \class{QMA\textsubscript{EXP}}-completeness due to (in part) exponential scaling between input length and system size.}.

\subsection{Sparse Languages and Promise Problems}
\begin{definition}[Sparse language] \label{def:sparse_l} \cite{mah82}
A language $L \subseteq \Sigma^{*}$ is called \emph{sparse} if there exists a polynomial $P$ such that for every $n \in \mathbb{N}$ there are at most $P(n)$ words of length at most $n$ in $L$.
\end{definition}

Inspired by the Definition~\ref{def:sparse_l} of sparse languages, we define sparse promise problems as follows:
\begin{definition}[Sparse promise problem] \label{def:sparse_pp}
A promise problem $\Pi = (\Pi_{Yes}, \Pi_{No})$ is called \emph{sparse} if the language $L \triangleq \Pi_{Yes}$ is sparse.
\end{definition}
Note that we could have defined sparse promise problems differently, for example by requiring the full promise $\Pi_{Yes} \cup \Pi_{No}$ to be a sparse language.
Our definition is chosen mainly because it straightforwardly generalizes the definition of sparse languages (namely, a language $L$ is sparse if and only if its corresponding promise problem $\Pi = (L, \overline{L})$ is sparse).
Another advantage of this choice is that Mahaney's and Fortune's theorems, described in Subsection~\ref{subsec:mah_for}, easily generalize to sparse promise problems.

\subsection{Mahaney's and Fortune's Theorems} \label{subsec:mah_for}
Using Definition~\ref{def:sparse_l}, in 1979 Fortune proved the following theorem (theorem 1 in \cite{for79}):
\begin{theorem}[Fortune] \label{thm:for}
If there exists a sparse language $L$ that is \class{coNP-Hard}, then $\class{P} = \class{NP}$
\end{theorem}
The proof given by Fortune shows a polynomial time algorithm for the language $SAT$, assuming a co-sparse \class{NP-Hard} language $L$ exists (a co-sparse language is a language $L$ whose complement $\bar{L}$ is sparse).

Following Fortune, in 1982 Mahaney showed the complementary result \cite{mah82}:
\begin{theorem}[Mahaney] \label{thm:mah}
If there exists a sparse language $L$ that is \class{NP-Hard}, then $\class{P} = \class{NP}$
\end{theorem}
Although the original proof of the theorem is quite complicated, a simplified proof based on a polynomial time algorithm for $SAT$ has been presented by Grochow in \cite{gro16}.

We later generalize the results of Fortune and Mahaney to promise problems in Subsection~\ref{subsec:mah_promise}.

\section{Results} \label{sec:res}
\subsection{Reducing 2D TI Heisenberg Problem to its Sparse Version} \label{subsec:red_to_sparse}
Our first step is reducing Problem~\ref{prob:heis} to its following ``sparse'' version:
\begin{problem}[Sparse two-dimensional (2D) translationally-invariant (TI) Heisenberg Hamiltonian]
\label{prob:sparse_heis}
    Defined similarly to Problem~\ref{prob:heis}, except that the input arguments $J, \alpha, \beta$ are each specified using $O(\log(N))$ bits.
    $N$ is still specified in unary, with $N$ bits.
\end{problem}
Thus we have for each $N$ at most $2^{O(\log(N))} = Poly(N)$ words in $\Pi_{Yes}$, due to a single choice for the argument $N$ (the unary $1^N$) and $2^{O(\log(N))}$ binary choices for the rest of the arguments.
Since the total input length is $\theta(N)$, with appropriate padding of the input arguments it can be shown Problem~\ref{prob:sparse_heis} is a sparse promise problem (by Definition~\ref{def:sparse_pp}).

For each of the integer and fractional parts of the input arguments $J, \alpha, \beta$, we reduce a problem instance with $\Omega(\log(N))$ bits to an equivalent problem instance with $O(\log(N))$ bits.
Thus the overall reduction composes the reductions given below a constant number of times, each time on a different argument.

\paragraph{Integer part of $J$}
According to the definition of Problem~\ref{prob:heis}, $J \in (0, P(N))$ for some polynomial $P$.
Thus the integer part of $J$ is trivially specified by at most $O(\log(P(N))) = O(\log(N))$ bits.

\paragraph{Fractional part of $J$}
Given an input Hamiltonian $H$ represented by the input arguments $N, J, \alpha, \beta$ where the fractional part of $J$ is specified by $M = \Omega(\log(N))$ bits,
define $J'$ as the truncation of $J$ to $L = k\log(N)$ bits in the fractional part (for some constant $k > 0$ to be determined later), with the integer part unchanged.
Since $J$ differs from $J'$ by at most all the $M-L$ least significant bits of the fractional part, we have:
\begin{equation}
    |J - J'| \leq \frac{1}{2^L} = \frac{1}{N^{k}}.
\end{equation}
Thus the overall difference of the Hamiltonian $H'$ (represented by the different input arguments $N, J'$) from the Hamiltonian $H$ in operator norm is:
\begin{equation}
    \|H - H'\| = |J-J'|\cdot\left\lVert\sum_{\nn{i,j}}\vec{\sigma}_{i}\cdot\vec{\sigma}_{j}\right\rVert \leq \frac{1}{N^{k}} \cdot a N^2 \triangleq \varepsilon,
\end{equation}
for some constant $a > 0$, where the last inequality is due to the operator norm of the Heisenberg Hamiltonian $\sum_{\nn{i,j}}\vec{\sigma}_{i}\cdot\vec{\sigma}_{j}$ scaling as the number of edges of the 2D lattice, $O(N^2)$.

By Weyl's Perturbation Theorem (Corollary III.2.6 in \cite{bha97}), we have:
\begin{equation}
    |\lambda(H)-\lambda(H')| \leq \|H - H'\| \leq \varepsilon.
\end{equation}
Therefore, reducing the instance of $H$ to an instance $H'$ with arguments $N, J', \alpha' \triangleq \alpha + \varepsilon, \beta' \triangleq \beta - \varepsilon$ is a valid reduction.
Moreover, $\alpha', \beta'$ satisfy:
\begin{equation}
     \beta' - \alpha' = \beta - \alpha - 2\varepsilon > \frac{1}{N^c} - \frac{2a}{N^{k-2}},
\end{equation}
with the last inequality due to the promise gap $\beta - \alpha > 1/N^{c}$.
Choosing $k=2c+2$ we obtain $\beta'-\alpha' > 1/N^{2c}$ (for all $N$ such that $N^c > 2a + 1$), satisfying an inverse-polynomial gap as necessary.
\footnote{Since $\varepsilon$ is inverse-polynomial in $N$, it cannot change the number of bits needed to represent $\alpha', \beta'$ compared to $\alpha, \beta$ by more than $O(\log(N))$.}

\paragraph{Integer parts of $\alpha, \beta$}
For any Hamiltonian $H$, by definition of the operator norm $\|H\|$ we can bound the ground state energy $\lambda(H)$ by:
\begin{equation}
    |\lambda(H)| \leq \|H\|.
\end{equation}
Due to this bound, any problem instance with $\alpha \leq -\|H\|$ cannot be a `Yes' instance, as $\lambda(H) < \alpha \leq -\|H\|$ is impossible.
Similarly, any problem instance with $\beta \geq \|H\|$ cannot be a `No' instance.
Moreover, the operator norm $\|H\|$ itself is upper-bounded by:
\begin{equation}
    \|H\| = J\cdot\left\lVert\sum_{\nn{i,j}}\vec{\sigma}_{i}\cdot\vec{\sigma}_{j}\right\rVert \leq J\cdot a N^2,
\end{equation}
for some constant $a > 0$.
Note the upper-bound $P \triangleq J\cdot a N^2$ is a polynomial function, due to $0 < J < Poly(N)$.

Thus, given an input Hamiltonian $H$ represented by the input arguments $N, J, \alpha, \beta$ where the integer part of $\alpha$ or $\beta$ are specified by $\Omega(\log(N))$ bits, the reduction algorithm is as follows:
\begin{enumerate}
    \item Compute $P \triangleq J\cdot a N^2$
    \item If $\alpha \leq -P$ (and therefore $\alpha \leq -\|H\|$): output a fixed `No' instance with $O(\log(N))$ bits in the integer parts of $\alpha, \beta$
    \item If $\beta \geq P$ (and therefore $\beta \geq \|H\|$): output a fixed `Yes' instance with $O(\log(N))$ bits in the integer parts of $\alpha, \beta$
    \item Otherwise, i.e. $\alpha, \beta \in (-P,P)$: output an instance with equal arguments $N, J$ but with the integer parts of $\alpha, \beta$ truncated to $O(\log(P)) = O(\log(N))$ bits
\end{enumerate}
We note that the reduction is valid for all input Hamiltonians that satisfy the promise ($\lambda(H) < \alpha$ or $\lambda(H) > \beta$).

\paragraph{Fractional parts of $\alpha, \beta$}
A reduction of the fractional parts of $\alpha, \beta$ to $O(\log(N))$ bits is obtained similarly to the reduction of the fractional part of $J$.
For a detailed description, see Appendix~\ref{sec:frac_ab}.

\subsection{Mahaney's and Fortune's Theorems for Promise Problems} \label{subsec:mah_promise}
Having seen that Problem~\ref{prob:heis} can be reduced (in polynomial time) to Problem~\ref{prob:sparse_heis} (Sparse 2D TI Heisenberg), the next step is a straightforward formulation of Mahaney's and Fortune's theorems (described in Subsection~\ref{subsec:mah_for}) for promise problems:
\begin{proposition}[Fortune's theorem for promise problems]\label{prop:for_promise}
If there exists a sparse promise problem $\Pi$ that is \class{coNP-Hard}, then $\class{P} = \class{NP}$
\end{proposition}
\begin{proposition}[Mahaney's theorem for promise problems]\label{prop:mah_promise}
If there exists a sparse promise problem $\Pi$ that is \class{NP-Hard}, then $\class{P} = \class{NP}$
\end{proposition}
The proof of both propositions follows:
\begin{proof}
Let $\Pi = (\Pi_{Yes}, \Pi_{No})$ be a sparse promise problem that is \class{(co)NP-Hard}.
Define a language $L = \Pi_{Yes}$.
Then $L$ is a sparse language by the definition of sparse promise problems (Definition~\ref{def:sparse_pp}).
Moreover, $L$ is \class{(co)NP-Hard}:
let $L' \in \class{(co)NP}$, then since $\Pi$ is \class{(co)NP-Hard} there exists a polynomial time function $f$ such that:
\begin{gather}
    x \in L' \Rightarrow f(x) \in \Pi_{Yes} = L, \\
    x \notin L' \Rightarrow f(x) \in \Pi_{No} \subseteq \bar{L}.
\end{gather}
Having shown the existence of a sparse language that is \class{(co)NP-Hard}, we can now apply Mahaney's (Fortune's) theorem to deduce $\class{P} = \class{NP}$
\end{proof}

\subsection{Conclusion of Claim~\ref{clm:main}} \label{subsec:conc}
With propositions~\ref{prop:for_promise}--\ref{prop:mah_promise} in place, we restate the main result:
\begin{repclaim}{1}[Restated]
    If $\class{P}\neq\class{NP}$, Problem~\ref{prob:heis} is neither $\class{NP-Hard}$ nor $\class{coNP-Hard}$.
\end{repclaim}
The proof of the claim is almost immediate:
\begin{proof}
Assume that Problem~\ref{prob:heis} is $\class{NP-Hard}$.
As proved in Subsection~\ref{subsec:red_to_sparse}, Problem~\ref{prob:heis} can be reduced in polynomial time to Problem~\ref{prob:sparse_heis}, meaning Problem~\ref{prob:sparse_heis} is $\class{NP-Hard}$ as well.
Since Problem~\ref{prob:sparse_heis} is a sparse promise problem that is \class{NP-Hard}, Proposition~\ref{prop:mah_promise} shows that $\class{P} = \class{NP}$.
Therefore, if Problem~\ref{prob:heis} is $\class{NP-Hard}$ then $\class{P} = \class{NP}$;
a similar proof shows that if Problem~\ref{prob:heis} is $\class{coNP-Hard}$, then $\class{P} = \class{NP}$.
\end{proof}

\subsection{\texorpdfstring{$d$}{d}-dimensional \texorpdfstring{$O(1)$}{O(1)}-Unit-Cell Hamiltonians} \label{subsec:uc-h}
We note that throughout the reduction described in Subsection~\ref{subsec:red_to_sparse}, we did not need the interaction to be of the specific form $\vec{\sigma}_i\cdot\vec{\sigma}_j$, but only needed its operator norm to be bounded by some constant $a = O(1)$.
Following this, we generalize the result to a much broader family of Hamiltonians, which we now define:

\begin{definition}[$d$-dimensional $O(1)$-Unit-Cell Hamiltonian] \label{def:uc-h}
A Hamiltonian $H$ is called $d$-dimensional $O(1)$-Unit-Cell if it is of the following form (on the $N^d$ square grid):
\begin{equation} \label{eq:uc-h}
	H=\sum_{i_1=1}^N\sum_{i_2=1}^N\dots\sum_{i_d=1}^N\widetilde{H}_{i_1,i_2,\ldots,i_d},
\end{equation}
where $\widetilde{H}$ is an arbitrary and \emph{constant} Hamiltonian on a ``Unit-Cell'' of $O(1)$ qubits:
\begin{equation}
    \widetilde{H}_{i_1,i_2,\ldots,i_d} = \sum_{Q\subseteq W_{i_1,i_2,\ldots,i_d}}J_Q \tilde{h}_Q,
\end{equation}
where $W_{i_1,i_2,\ldots,i_d}$ is the set of $O(1)$ qubits in the unit-cell indexed by $i_1,i_2,\ldots,i_d$; $J_Q$ are scalars with $|J_Q| = O(Poly(N))$; and $\tilde{h}_Q$ is a Hamiltonian on subset $Q$ with $\|\tilde{h}_Q\|=O(1)$.
\end{definition}
We emphasize that the Hamiltonians $\widetilde{H}_{i_1,i_2,\ldots,i_d}$ and $\widetilde{H}_{i'_1,i'_2,\ldots,i'_d}$ are of equal form, save for acting on different unit-cells.
We also note that the unit-cells are not necessarily disjoint, i.e. there may be qubits in the full Hamiltonian participating in more than one unit-cell.
For examples of unit-cells, see Figure~\ref{fig:uc}.

With Definition~\ref{def:uc-h} in place, we can define a local Hamiltonian ground state energy problem analogous to Problem~\ref{prob:heis} (for constant $d$), with $N$ specified in unary and $\widetilde{H}$ in binary, and claim the following:
\begin{claim} \label{clm:uc-h}
    If $\class{P}\neq\class{NP}$, the Hamiltonian promise problem of eq.~\eqref{eq:uc-h} is neither \class{NP-Hard} nor \class{coNP-Hard}.
\end{claim}
\begin{proof}
We define a ``sparse'' version of the problem, with $\alpha, \beta$ given by $O(\log(N))$ bits and $\widetilde{H}$ given by $O(\log(N))$ bits (of which $O(\log(N))$ bits are needed to specify the $J_Q$ scalars, and $O(1)$ bits are needed to specify the $\tilde{h}_Q$ hamiltonians, such that the Hamiltonian in eq.~\eqref{eq:uc-h} may be reconstructed uniquely).
We note that $\widetilde{H}$ consists of $O(1)$ terms, since the power set of an $O(1)$-sized set unit-cell is still an $O(1)$-sized set.
Each of the $O(1)$ scalars $J_Q$ is polynomially large, and thus its integer part is trivially specified by $O(\log(N))$ bits.
Reducing the fractional part of $J_{Q_0}$ (for a \emph{single, specific} subset $Q_0$ of $W_{i_1,i_2,\ldots,i_d}$) to a $k\cdot\log(N)$-bit truncation $J_{Q_0}'$ (with corresponding Hamiltonian $H'$) follows analogously to the reduction of $J$ in Subsection~\ref{subsec:red_to_sparse};
holding the other arguments constant (including $J_Q'=J_Q$ for all $Q \neq Q_0$), and assuming $\|\tilde{h}_Q\| \leq a = O(1)$ for all $Q$, we have:
\begin{equation}
\begin{split}
    \|H-H'\| &= \left\lVert\sum_{i_1=1}^{N}\sum_{i_2=1}^{N}\dots\sum_{i_d=1}^{N}\sum_{Q\subseteq W_{i_1,i_2,\ldots,i_d}}(J_Q-J_Q')\tilde{h}_Q\right\rVert \\
    &= \left\lVert\sum_{i_1=1}^{N}\sum_{i_2=1}^{N}\dots\sum_{i_d=1}^{N}(J_{Q_0}-J_{Q_0}')\tilde{h}_{Q_0}\right\rVert \leq |J_{Q_0}-J_{Q_0}'| \cdot \|\tilde{h}_{Q_0}\| \cdot N^d \leq \frac{1}{N^{k}} \cdot a \cdot N^d.
\end{split}
\end{equation}
The reduction process is repeated for each $J_{Q}$ separately, and since the number of consecutive reductions is still constant, their composition is a polynomial reduction.
The rest of the proof follows identically to Subsections~\ref{subsec:red_to_sparse} and \ref{subsec:conc} (except that we choose $k = 2c + d$ instead of $k = 2c + 2$).
\end{proof}

We note that the Hamiltonian described by Definition~\ref{def:uc-h} includes many important translationally-invariant Hamiltonians as special cases.
In particular, the previously discussed Heisenberg Hamiltonian (eq.~\eqref{eq:heis}) with closed boundary conditions is a special case, with $d=2$, and $\widetilde{H}$ being a ``star'' of 5 qubits and 4 edges (Figure~\ref{fig:int_uc}).
Other important models included in this family are the translationally-invariant (with closed boundary conditions) Fermi-Hubbard and t-J Hamiltonians (Subsection~\ref{subsec:hubb}).

\begin{figure}[!ht]
	\centering
	\subfloat[2D Lattice Unit-Cell\label{fig:int_uc}]{%
        \includegraphics[width=0.14\textwidth]{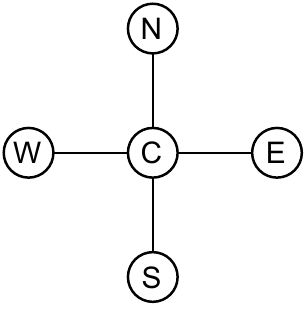}
	}
	\hspace{0.15\textwidth}
    \subfloat[Complex Unit-Cell\label{fig:complex_uc}]{%
        \includegraphics[width=0.14\textwidth]{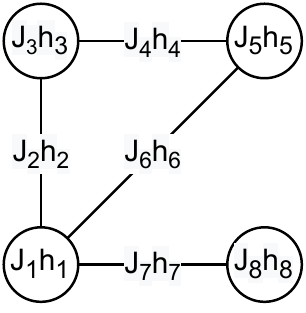}
    }
    \caption{(a) The graph of the closed-boundary 2D Heisenberg, Fermi-Hubbard and t-J Hamiltonians unit-cell.
    For the Heisenberg model (eq.~\eqref{eq:heis}), the Hamiltonian is:
    $\widetilde{H} = (J/2) \sum_{q\in A} \vec{\sigma}_C\cdot\vec{\sigma}_q$ and $A = \{N, E, S, W\}$.
    Each edge participates in 2 unit-cells, demonstrating that the unit-cells are not disjoint in this case.
    (b) An example of a more complicated Unit-Cell.
    Each qubit/edge has its own Hamiltonian $h_i$ and scalar $J_i$.}
    \label{fig:uc}
\end{figure}

Lastly, we note it is unknown whether the same claim is true for unit-cells of size $O(\log(N))$, since then we would na\"ively obtain a problem specified by an input of $O(\log^3(N))$ bits ($O(\log(N))$ vertices in the unit-cell, with at most $O(\log^2(N))$ edges, each edge associated with an interaction strength specified by $O(\log(N)$ bits).
Such a problem is of course not sparse by the current definition.
Thus to potentially prove such a result, one may either attempt to generalize Mahaney's and Fortune's theorems to include problems with such scaling or analyze more carefully the Hamiltonian in an attempt to establish tighter bounds.

\subsection{Two-dimensional (2D) translationally-invariant (TI) Fermi-Hubbard Hamiltonian} \label{subsec:hubb}
We can apply our result of Subsection~\ref{subsec:uc-h} to include models such as the 2D TI Fermi-Hubbard Hamiltonian.
The 2D TI Fermi-Hubbard Hamiltonian is a fermionic Hamiltonian and a generalization of the Heisenberg Hamiltonian defined in eq.~\eqref{eq:heis}.
On an $N \times N$ lattice, it is of the form:
\begin{equation} \label{eq:hubb}
    H_{\textrm{TI-Fermi-Hubbard}} = -t\sum_{\nn{i,j},s\in\{\uparrow,\downarrow\}}(a_{i,s}^{\dagger}a_{j,s}+a_{j,s}^{\dagger}a_{i,s})+U\sum_{i} n_{i,\uparrow}n_{i,\downarrow},
\end{equation}
where $t$ and $U$ are real scalars, and $a_{i,s}, a^{\dagger}_{i,s}$, $n_{i,s} = a^\dagger_{i,s}a_{i,s}$ are the fermionic annihilation, creation and number operators for spin $s \in \{\uparrow, \downarrow\}$ on lattice site $i$, respectively.
The number of electrons $0 < N_e < 2N^2$ is also specified with the Hamiltonian.
We note that in the case $U \gg t$ and $N_e = N^2$ (Half-filling regime), the Fermi-Hubbard Hamiltonian is reduced to the antiferromagnetic Heisenberg Hamiltonian.

We claim a similar result holds for the Fermi-Hubbard model:
\begin{claim} \label{clm:hubb}
    If $\class{P}\neq\class{NP}$, the Hamiltonian promise problem of eq.~\eqref{eq:hubb} is neither \class{NP-Hard} nor \class{coNP-Hard}.
\end{claim}
\begin{proof}
Firstly, we note that since $0 < N_e < 2N^2$ it is already specified by $O(\log(N))$ bits by definition.
The rest of the result follows as a corollary of Claim~\ref{clm:uc-h}.
Taking the unit-cell as described in Figure~\ref{fig:int_uc}, with:
\begin{equation}
    \widetilde{H} = -\frac{t}{2}\sum_{q\in A, s\in\{\uparrow,\downarrow\}}(a_{C,s}^{\dagger}a_{q,s}+a_{q,s}^{\dagger}a_{C,s})+U n_{C,\uparrow}n_{C,\downarrow},
\end{equation}
where $C$ is the center qubit of the unit-cell, and $A = \{N, E, S, W\}$.
We thus obtain that the Fermi-Hubbard Hamiltonian problem of eq.~\eqref{eq:hubb} is a special case of the 2-dimensional $O(1)$-Unit-Cell Hamiltonian problem of Definition~\ref{def:uc-h}.
\end{proof}
A similar result applies to the translationally-invariant t-J model, an intermediate model between Heisenberg and Fermi-Hubbard, as proved in Appendix~\ref{sec:tj}.

\section{Discussion} \label{sec:disc}
In this paper, we have looked at the problem of finding the ground state energy of a specific, yet widely studied Hamiltonian---the translationally-invariant Heisenberg Hamiltonian on the 2D square lattice.
Using complexity theoretic tools we were able to show that the problem, stated as a local Hamiltonian promise problem, cannot be \class{NP-Hard} or \class{coNP-Hard} under the assumption of $\class{P} \neq \class{NP}$.
We then generalized this result to a much larger family of translationally-invariant \emph{unit-cell} Hamiltonians, and to a translationally-invariant version of more general models---the Fermi-Hubbard and t-J models.

\paragraph{Potential positive results}
Our non-membership result, along with the already proven membership of the problem in \class{StoqMA}, raises questions about potential positive results for the problem.
Despite many possible options, we discuss three that we deem to be consequential.

The first possibility is that the problem is shown to be in \class{P} or \class{BPP}.
Since exact diagonalization of the Hamiltonian or even finding (analytically) its ground state energy still eludes scientists to this day, such a result would significantly improve our understanding of problems in condensed matter physics.

Another possibility is that the problem is shown to be in \class{BQP} (and not in \class{P} or \class{BPP}).
This possibility would hint at quantum supremacy, as it would show the problem to be efficiently solvable by quantum computers, but not by classical computers.

Lastly, the problem could be in \class{NP}.
Under the assumptions of $\class{P} \neq \class{NP}$ and that no \class{P} algorithm exists, the problem would thus be \class{NP-Intermediate}.
Although not proven, some noticeable problems like graph isomorphism and integer factorization are conjectured to be \class{NP-Intermediate}, with the latter even belonging to \class{BQP} due to Shor's algorithm \cite{shor99}.
Hence the problem could be an unusual example of an \class{NP-Intermediate} problem coming from the physical world of local Hamiltonian problems.

\paragraph{Potentially tractable special cases of \class{NP-Complete} problems}
The Ising model, a ``classical'' counterpart to the Heisenberg model (see Appendix~\ref{sec:ising} for details), was studied by Barahona in 1982 using a complexity-theoretic approach \cite{bar82}.
It was shown by \cite{bar82} that when the Ising Hamiltonian is restricted to a 2D square lattice and no 1-local terms, the problem of finding the ground state (and its energy) is in \class{P}, using Edmonds' algorithm from matching theory \cite{edm65}.
Barahona also showed, however, that when restricted to a 2D square lattice \emph{with} arbitrary 1-local terms, or when restricted to a 3D cubic lattice without 1-local terms, finding the ground state energy becomes an \class{NP-Complete} problem.

We note that using the results of Subsection~\ref{subsec:uc-h}, one may show that large subclasses of such Hamiltonians, i.e. Ising Hamiltonians with ``Unit-Cell'' translational-invariance, cannot, in fact, be \class{NP-Hard}.
Since many known \class{NP-Complete} problems can be formulated in terms of Ising Hamiltonian problems \cite{luc14}, such results may potentially lead to certain special cases of these problems being excluded from \class{NP-Complete}.

\paragraph{Possible further generalizations}
The techniques we used, and especially the reduction to a ``sparse'' version of the problem in Subsection~\ref{subsec:red_to_sparse}, could potentially be applied to similar Hamiltonian problems.
In particular, we conjecture the result could be further generalized to any local Hamiltonian problem specified by: a single system size related argument $N$ (specified in unary), $O(1)$ additional arguments of the Hamiltonian (each at most polynomially large in $N$, specified in binary) and energy threshold arguments $\alpha, \beta$ (specified in binary) satisfying an inverse-polynomial promise gap.

A similar conjecture can be made with $O(\log(N))$ additional arguments of the Hamiltonian (instead of $O(1)$).
In such a case, however, our techniques do not na\"ively apply (as mentioned in Subsection~\ref{subsec:uc-h}), and thus we leave the exploration and proof or refutation of this conjecture to future work.

\appendix

\section{Membership of Problem~\ref{prob:heis} in S\lowercase{toq}MA} \label{sec:stoqma}
In Subsection~\ref{subsec:prev_work}, we have argued the following:
\begin{fact} \label{fac:stoqma}
Problem~\ref{prob:heis} is in \class{StoqMA}.
\end{fact}
The class \class{StoqMA} was first defined in \cite{bbt06}.
Besides their original definition of \class{StoqMA}, which extends an alternative definition of \class{MA}, \cite{bbt06} also showed that the local Hamiltonian problem for ``Stoquastic'' Hamiltonians is \class{StoqMA-Complete}.
A Stoquastic Hamiltonian is one for which the off-diagonal terms in the standard basis are real and non-positive.
For completeness, we now provide a short proof of Fact~\ref{fac:stoqma} (based on \cite{cm16}):
\begin{proof}
For a bipartition of the 2D lattice vertices $V=A\cup B$, apply a unitary transformation $U=\bigotimes_{a\in A}\mathbb{I}_{a}\otimes\bigotimes_{b\in B}Z_{b}$, which transforms $H_{\text{TI-Heisenberg}}$ of eq.~\eqref{eq:heis} into:
\begin{equation}
    H_{\text{Stoq}}=UH_{\text{TI-Heisenberg}}U^\dagger=J\sum_{\nn{a,b}}(-X_{a}X_{b}-Y_{a}Y_{b}+Z_{a}Z_{b})
\end{equation}
As unitary transformations do not change the spectrum of the Hamiltonian, the ground state energy of $H_{\text{Stoq}}$ is the same as of $H_{\text{TI-Heisenberg}}$.
$H_{\text{Stoq}}$ has non-positive off-diagonal elements (stoquastic form) and thus its ground state energy problem is in \class{StoqMA}, as shown by \cite{bbt06}.
Overall the problem is reduced to a problem in \class{StoqMA}.
\end{proof}

\section{The Ising Hamiltonian} \label{sec:ising}
A well-studied and highly related model to the Heisenberg Hamiltonian is the Ising Hamiltonian (see \cite{cip87} for an introductory review).
The Ising Hamiltonian is in a sense a ``classical'' version of the Heisenberg Hamiltonian, given in the general case by:
\begin{equation} \label{eq:ising}
	H_{\text{Ising}}=\sum_{i,j}J_{ij}Z_{i}Z_{j}+\sum_{i}B_{i}Z_{i}.
\end{equation}
As the Hamiltonian is diagonal, its ground state is a computational basis state.
Thus local Hamiltonian problems derived from eq.~\eqref{eq:ising} and its special cases are in the complexity class \class{NP}, with the hint being the ground state encoded as a bit string.
While an exact solution to the general Ising model is unknown, several special cases have been analyzed.
In 1925, Ising (after whom the model is named) solved the 1D model exactly \cite{isi25}.

\section{Fractional parts of \texorpdfstring{$\alpha, \beta$}{a, b} reduction} \label{sec:frac_ab}
We now describe in detail how to reduce the fractional parts of $\alpha, \beta$ of Problem~\ref{prob:heis} to $O(\log(N))$ bits.

Given an input Hamiltonian $H$ represented by the input arguments $N, J, \alpha, \beta$ where the fractional parts of $\alpha, \beta$ are specified by $M = \Omega(\log(N))$ bits, it follows by the promise gap that for some constant $c>0$:
\begin{equation} \label{eq:frac_ab_promise}
    \beta - \alpha > \frac{1}{N^{c}}.
\end{equation}
Let $L = k\cdot\log(N)$ for some $k > c$ to be determined later, and define $\alpha', \beta'$ as the truncations of $\alpha, \beta$ to $L = O(\log(N))$ bits in the fractional part, with the integer parts unchanged.
Thus we have for $\alpha'$ and $\beta'$:
\begin{gather}
    |\alpha - \alpha'| \leq \frac{1}{2^L} \label{eq:frac_ab_alpha}, \\ 
    |\beta - \beta'| \leq \frac{1}{2^L} \label{eq:frac_ab_beta}.
\end{gather}
Define $\alpha'' \triangleq \alpha' + 1/2^L$ and $\beta'' \triangleq \beta' - 1/2^L$.
We note that the fractional parts of $\alpha'', \beta''$ still require only $O(L) = O(\log(N))$ bits to represent.
Moreover, we have due to eqs.~\eqref{eq:frac_ab_promise}--\eqref{eq:frac_ab_beta}:
\begin{equation}
    \beta'' - \alpha'' = \beta' - \alpha' - \frac{2}{2^L} \geq \beta - \alpha - \frac{4}{2^L} > \frac{1}{N^c}-\frac{4}{N^k}.
\end{equation}
Thus, by choosing $k = 2c$ we obtain $\beta''-\alpha'' > 1/N^{2c}$ (for all $N$ such that $N^c > 5$), satisfying an inverse-polynomial gap.
The reduction therefore outputs the given instance with $\alpha'', \beta''$ instead of $\alpha, \beta$.
Validity of the reduction follows from eqs.~\eqref{eq:frac_ab_alpha}--\eqref{eq:frac_ab_beta}.

\section{Two-Dimensional (2D) translationally-invariant (TI) \lowercase{t}-J Hamiltonian} \label{sec:tj}
The t-J model arises from the Fermi-Hubbard model under the $U \gg t$ constraint, but when the regime is not half-filling (namely, $N_e \neq N^2$).
The Hamiltonian in the two-dimensional, translationally-invariant case is \cite{pla02}:
\begin{equation} \label{eq:tj}
	H_{\text{TI-t-J}} = -t\sum_{\nn{i,j},s\in\{\uparrow,\downarrow\}}(\tilde{a}^\dagger_{i,s}\tilde{a}_{j,s}+\tilde{a}^\dagger_{j,s}\tilde{a}_{i,s}) + J\sum_{\nn{i,j}}(\vec{S}_i\cdot\vec{S}_j-\frac{\tilde{n}_i \tilde{n}_j}{4}),
\end{equation}
where $t$ and $J$ are real scalars, $\tilde{a}_{i,s} = a_{i,s}(1 - n_{i,-s}), \tilde{a}^\dagger_{i,s} = (1 - n_{i,-s})a^\dagger_{i,s}$ are the projected annihilation/creation operators for a particle on site $i$ with spin $s \in\{\uparrow,\downarrow\}$ ($-s$ denoting the spin complementary to $s$); $\tilde{n}_i = \sum_{s} \tilde{a}^\dagger_{i,s}\tilde{a}_{i,s}$ is the projected number operator on site $i$ and $\vec{S}_i = \frac{1}{2}(\tilde{\sigma}^{x,i}, \tilde{\sigma}^{y,i}, \tilde{\sigma}^{z,i})$ is the spin operator on site $i$, with $\tilde{\sigma}^{\alpha,i} = \sum_{s,s'} \sigma^\alpha_{ss'}\tilde{a}^\dagger_{i,s}\tilde{a}_{i,s'}$.
As in the Fermi-Hubbard model, the number of electrons $N_e$ is specified along with the Hamiltonian.
\begin{claim} \label{clm:tj}
    If $\class{P}\neq\class{NP}$, the Hamiltonian promise problem of eq.~\eqref{eq:tj} is neither \class{NP-Hard} nor \class{coNP-Hard}.
\end{claim}
\begin{proof}
The proof is analogous to the proof of Claim~\ref{clm:hubb}, with the unit-cell:
\begin{equation}
    \widetilde{H} = -\frac{t}{2}\sum_{q\in A, s\in\{\uparrow,\downarrow\}}(\tilde{a}^\dagger_{C,s}\tilde{a}_{q,s}+\tilde{a}^\dagger_{q,s}\tilde{a}_{C,s}) + \frac{J}{2}\sum_{q\in A}(\vec{S}_C\cdot\vec{S}_q-\frac{\tilde{n}_C \tilde{n}_q}{4}).
\end{equation}
\end{proof}
\bibliographystyle{ACM-Reference-Format}
\bibliography{not_nph}

\end{document}